\documentclass[preprint,12pt,fleqn]{elsarticle}

\long\def\comment#1{}

\newcommand{\be}{\begin{equation}}
\newcommand{\ee}{\end{equation}}
\usepackage{amsthm}
\newtheorem{Thm}{Theorem}          
\newtheorem{lemma}[Thm]{Lemma}   
\newtheorem{definition}{Definition}
\newtheorem{remark}{Remark} 

\newtheorem*{proof*}{Proof}
\newtheorem{assumption}{Assumption}

\newfont{\bbb}{msbm10 scaled 700}

\newfont{\bb}{msbm10 scaled 1100}

\newcommand{\RNum}[1]{\uppercase\expandafter{\romannumeral #1\relax}}

\usepackage{dsfont}
\usepackage{amsmath}

\usepackage{circuitikz}
\usepackage{booktabs}
\usepackage{amssymb}
\usepackage{algorithm}
\usepackage{algpseudocode}

\usepackage{graphicx}
\usepackage{epstopdf}
\usepackage{multirow}
\usepackage{cite}
\usepackage{IEEEtrantools}

\usepackage{subcaption}
\usepackage{placeins}
\usepackage{float}
\usepackage{xcolor}
\usepackage{microtype}
\usepackage[hidelinks]{hyperref}

\journal{Journal of the Franklin Institute}

\begin{document}

\begin{frontmatter}

\title{Stochastic MPC under Heavy-Tailed Disturbances: An Extreme Value Theory Approach}

\author[ncsu]{Xiuzhen Ye}
\author[ncsu]{Wentao Tang\corref{cor1}}
\ead{wtang23@ncsu.edu}
\cortext[cor1]{Corresponding author. Department of Chemical and Biomolecular Engineering, North Carolina State University, 911 Partners Way, Engineering Building I, Raleigh, NC 27695-7905, United States.}

\affiliation[ncsu]{organization={Department of Chemical and Biomolecular Engineering, North Carolina State University},
    addressline={911 Partners Way, Engineering Building I},
    city={Raleigh},
    state={NC},
    postcode={27695-7905},
    country={United States}}

\begin{abstract}
Safety-critical control systems must contend with disturbances whose extreme deviations occur far more frequently than classical light-tailed models predict. Existing stochastic MPC (SMPC) formulations tighten constraints using an assumed distribution, a moment bound, or a finite scenario sample, each of which degrades under an unknown heavy-tailed disturbance. This paper develops an SMPC formulation for linear systems under heavy-tailed disturbances that is only assumed to be regularly varying, which allows an explicit extreme value theory (EVT) characterization of the tube error tail that is asymptotically exact. We further show that closed-loop dynamics induce temporal clustering of rare excursions across the prediction horizon, and characterize this clustering through a closed-form extremal index estimable from data. The resulting clustering-corrected constraint bounds the probability of a rare-event episode over the horizon, rather than only the marginal per-step exceedance probability. Simulations of a linear two-CSTR-in-series process and of a nonlinear unicycle navigating past an obstacle under Student-$t$ disturbances validate both approaches and demonstrate reduced frequencies of safety constraint violations.
\end{abstract}



\end{frontmatter}

\section{Introduction}
Safety-critical control systems must account for disturbances stemming from a variety of sources, including sensor faults, actuator glitches, and other rare operational upsets. When these disturbances are random variables, their extreme deviations can be unbounded, and if their distributions are heavy-tailed, the extreme deviations decay polynomially rather than exponentially in magnitude, which is much more severe than the benign case of light-tailed disturbances~\citep{Resnick_2007}. While a nominal controller that ignores disturbance risk offers no guarantee against rare excursions, guaranteeing robustness in a strictly deterministic sense, as in robust MPC~\citep{TubeMPC_05}, is usually infeasible in this setting, since no finite bound on the disturbance exists. Chance-constrained stochastic MPC (SMPC)~\citep{Mesbah2016_overview, FarinaGiulioniScattolini2016} is the standard way of balancing these two extremes, enforcing safety with a specified low violation probability (rather than for every possible realization). The theoretical foundations of chance-constrained SMPC are established typically by proving the measurability of closed-loop trajectories, probabilistic recursive feasibility, and consequently, asymptotic stability in expectation or in probability, see, e.g.,~\citep{McAllisterRawlings2022}.
 
Existing chance-constrained SMPC formulations handle the disturbance distribution in essentially four ways, and each relies on assumptions whose validity is not guaranteed under a genuinely heavy-tailed regime. The four approaches are (i) an analytic, distribution-specific tightening, (ii) a moment-based tightening using Cantelli--Chebyshev-type inequalities, (iii) a scenario-based approach that approximates the chance constraint through sampling, and (iv) a calibration-based approach that uses conformal prediction to tighten the constraint from held-out data.
In the first approach, the disturbance distribution is assumed to be Gaussian or otherwise fully specified, and the chance constraint on a linear combination of the state is tightened by inverting the assumed distribution's cumulative distribution function at the target violation probability, so that the size of the tightening results directly from the tail of the assumed distribution. The tube-based mixed stochastic-deterministic SMPC of~\citep{Mesbah_JPC_19} is a hybrid instance of this approach, splitting the uncertainty into a bounded, deterministic component handled by a classical robust reachable-set tube~\citep{TubeMPC_05}, that is, a fixed set guaranteed to contain the state under every disturbance realization, and a stochastic component with a known distribution whose chance constraint is tightened the same way. This approach is the most tractable, but it ties the guarantee to the assumed distribution and degrades if the disturbance is misspecified, or unknown a priori. 
 
The second approach relies on Cantelli--Chebyshev-type inequalities~\citep{MarshallOlkin1979}, which bound the probability that a random variable deviates far from its mean using only the mean and variance, with no assumption on the shape of the distribution. The joint chance-constrained formulation of Paulson et al.~\citep{PaulsonIJC17} follows this moment-based approach, building on the distributionally robust reformulation of Calafiore \& El Ghaoui~\citep{CalafioreElGhaoui2006}, that is, a reformulation guaranteed to hold under the worst distribution consistent with the known mean and variance, valid for arbitrary disturbance distributions provided the variance is finite.
This worst-case tightness comes at a cost, since the required back-off from the constraint boundary, needed to guarantee a violation probability of at most $\epsilon$, grows like $\sigma/\sqrt{\epsilon}$ (as $\epsilon \to 0$). This is far more conservative than the $\sigma\sqrt{\log(1/\epsilon)}$ back-off achieved when the full distribution is known, for example when the disturbance is Gaussian, meaning the moment-only bound demands a substantially larger safety margin to guarantee the same low violation probability. This limitation persists even in recent extensions to nonlinear systems, whose reachable-set constructions for unbounded process noise still require an assumed bound on the disturbance covariance~\citep{KohlerZeilinger2025}. Furthermore, the approach is clearly invalid when the disturbance is heavy-tailed, which we consider as the setting in the present paper.

The third, largely computational approach formulates the chance-constrained problem as a large-scale stochastic program and approximates it via scenario sampling combined with parallel decomposition~\citep{KumarZavala2019, JalvingShinZavala2022, KimPetraZavala2019}, where each sampled disturbance realization becomes one scenario in the optimization and the resulting large program is split across multiple computing nodes to keep the solution time manageable. Such methods accommodate arbitrary disturbance distributions in principle and scale to large systems, but their accuracy is fundamentally limited by the number of scenarios sampled, since rare events are, by construction, underrepresented in any finite scenario set. Moreover, for heavy-tailed distributions without a properly defined expectation, the law of large numbers may not even hold, which invalidates the scenario sampling approach~\citep{Durrett2019}. Another line of work uses conformal prediction to obtain distribution-free, finite-sample coverage guarantees for constraint tightening via calibration data rather than scenario optimization~\citep{Lindemann2023, Lindemann2024}, but these guarantees are marginal and quantile-based, and likewise do not explicitly account for heavy-tailed nonconformity scores or temporal clustering of rare events.

The gap left by these approaches is precisely the regime in which the disturbance is heavy-tailed, so that no assumed form of distribution, no finite moment, and no computationally tractable number of samples necessarily characterizes the rare-event risk captured in the tail of its distributions. Extreme Value Theory (EVT) is the branch of probability theory specifically concerned with the tail behavior of a distribution~\citep{Coles2001}, and it closes this gap directly by characterizing risk through the tail index $\alpha$, which describes how quickly the probability of an extreme deviation decays as that deviation grows large. A smaller $\alpha$ means a heavier tail in which large deviations remain comparatively likely, and a larger $\alpha$ means a lighter tail in which they become rare very quickly. This tail index remains well defined as long as the disturbance is i.i.d.\ and ``regularly varying'', a considerably weaker requirement than a fully specified distribution or a finite variance. To the best of our knowledge, this tail index characterization has not previously been used to design the constraint tightening in chance-constrained MPC. EVT has, however, been utilized in stochastic decision making, including risk analysis for general stochastic systems~\citep{ArsenaultChapman2022} and tail value estimation in risk-averse reinforcement learning~\citep{SomayajiLi2024}. 

The tail index characterization yields two advantages over the approaches surveyed above. First, the resulting quantile characterization is asymptotically exact as the target violation probability $\epsilon \to 0$, rather than conservative like the moment-based bound, or inaccurate like the scenario-based approach. Second, the tail characterization can be computed offline and built into a fixed constraint tightening, following the same tube-based approach as classical robust and stochastic MPC, so the resulting online optimization retains the computational structure of nominal linear MPC. Moreover, the tightening is computed directly from data, rather than from an assumed distribution or a moment bound, as in the sample-based tightening of~\citep{Lorenzen2017} and the scenario-based probabilistic reachable sets of~\citep{HewingZeilinger2020}.  
Neither approach models the disturbance beyond the observed samples, and consequently, their guarantees hold only at a finite-sample confidence level (rather than exactly), and are not tailored to a heavy-tailed regime. 
In contrast, the EVT-based tail index characterization in this work models the rate of tail decay parametrically and learns the tail decay parameter from data directly. The resulting quantile therefore extrapolates beyond the largest disturbance realized in the data.

Therefore, this paper develops an SMPC formulation that accounts for disturbances that are regularly varying with a heavy tail, replacing the invariant-set or moment-based bound on the tube error with an explicit, EVT quantile characterization. We further show that, even under i.i.d. disturbances, the closed-loop dynamics induce temporal clustering of rare excursions across the prediction horizon, and we characterize this clustering using Leadbetter's extremal index that can be estimated from data. The resulting formulation retains the computational structure, and hence the tractability, of nominal linear MPC.  

The contributions of this paper are as follows. First, we characterize the tails of the closed-loop state and of the tube prediction error under heavy-tailed disturbances using extreme value theory, and use the error tail to tighten the constraints of a tube-based stochastic MPC, so that rare excursions are constrained with a per-step violation probability that is asymptotically exact as the target becomes small. Second, we derive a closed-form extremal index for the closed-loop output along the constraint direction, which captures the clustering of rare excursions at consecutive time steps under i.i.d. heavy-tailed disturbances. Third, we use this extremal index to construct a clustering-corrected constraint that controls the risk of a rare-excursion episode over the prediction horizon, asymptotically as the horizon grows. The proposed approach is validated on a linear two-CSTR-in-series process and on a nonlinear unicycle navigating past an obstacle under Student-$t$ disturbances. Compared with Gaussian tightening, the EVT tightening reduces the realized rate of constraint violation by a factor of about 6 for the two-CSTR-in-series process and about 10 for the unicycle. The clustering correction further reduces this rate by a factor of 8 for the two-CSTR-in-series process and about 1.1 for the unicycle.  
  
\section{Preliminaries}\label{sec_background}
Consider a system with the following discrete-time linear dynamics
\begin{align}\label{eq_sys}
    x_{k+1} = & A x_k + B u_k + w_k, \quad 
y_k = c^\top x_k + d^\top u_k,
\end{align}
where $x_k \in \mathbb{R}^{n_x}$ is the state of the system; $u_k \in \mathcal{U} \triangleq \{ u \in \mathbb{R}^{n_u}: G u \leq g\}$ is the manipulated control input with hard constraints; and $y_k \in \mathbb{R}$ is the output. An i.i.d. sequence of random disturbances $\{w_k\}_{k\in\mathbb{N}}$ is introduced. In later discussions, we will consider only the case of $d=0$, without loss of generality.
 \subsection{Tube-based and stochastic MPC}\label{sec_tube_smpc}
Robust MPC confines the disturbance $w_k$ to a known compact set $\mathcal{W}$ and enforces the constraints for every realization in $\mathcal{W}$~\citep{BemporadMorari1999}. In tube-based robust MPC~\citep{TubeMPC_05}, a fixed stabilizing gain $K$ separates the state into a nominal trajectory and an error, and the error after $i$ steps lies in the reachable set $\mathcal{R}^e_i = \bigoplus_{j=0}^{i-1} A_K^j \mathcal{W}$, where $A_K \triangleq A - BK$ and $\oplus$ denotes the Minkowski sum. The constraints on the nominal trajectory are then tightened by $\mathcal{R}^e_i$, which is possible only when $\mathcal{W}$ is bounded. Stochastic MPC (SMPC)~\citep{Mesbah2016_overview} instead requires the state constraints to hold with a specified probability, so that for all $k \geq 1$,
\begin{equation}\label{eq_chance_const}
    \mathbb{P}( H x_k \leq h) \geq 1 - \beta,
\end{equation}
where $H \in \mathbb{R}^{r \times n_x}$, $h \in \mathbb{R}^r$, and $\beta \in [0,1]^r$ is a vector containing the allowed violation probability for $r$ state constraints. We let $\mathcal{X} \triangleq \{x \in \mathbb{R}^{n_x} : Hx \leq h\}$ and $\mathcal{U}$ defined in~\eqref{eq_sys} denote the state and input constraint sets, respectively. The existing ways of tightening~\eqref{eq_chance_const} are reviewed in Section~1, and none of them targets a heavy-tailed disturbance.

\subsection{Extreme value theory and regular variation}
\begin{definition}[Regular variation~\citep{Resnick_2007}]\label{def_reg_var}
A random variable $X$ is \emph{regularly varying with tail index $\alpha>0$} if
\[
\mathbb{P}(|X|>x) \sim L(x)\,x^{-\alpha}, \qquad x \to \infty,
\]
where the function $L$ is \emph{slowly varying}, i.e., $L(tx)/L(x) \to 1$ as $x\to\infty$ for every fixed $t>0$. The term $x^{-\alpha}$ captures the dominant polynomial decay rate of the tail probability, while $L$ is a correction that allows for general heavy-tailed distributions without forcing an exact power law.  
\end{definition} 
Distributions with $\alpha<2$ have infinite variance, and those with $\alpha\le1$ have infinite mean. The Student-$t$ distribution with $\nu$ degrees of freedom is regularly varying with tail index $\alpha=\nu$~\citep{Resnick_2007}, and it is the disturbance model used in Section~\ref{sec_simulation}.

Definition~\ref{def_reg_var} describes how fast the tail probability decays but not whether the extremes are large and positive or large and negative. This sign information is captured by the tail balance parameters.
\begin{definition}[Tail balance~\citep{Resnick_2007}]\label{def_tail_balance}
A regularly varying scalar random variable $\zeta$ with tail index $\alpha$ has \emph{tail-balance parameters} $(p,q)$, $p,q\ge0$, $p+q=1$, if
\begin{align}
\nonumber
& \mathbb{P}(\zeta>x) \sim p\,\mathbb{P}(|\zeta|>x), \\
\nonumber
& \text{and} \quad \mathbb{P}(\zeta\le-x) \sim q\,\mathbb{P}(|\zeta|>x), \qquad x\to\infty.
\end{align}
\end{definition}
 
 \subsection{Problem Statement}\label{sec_prelim}
In system~\eqref{eq_sys}, we now suppose that the disturbance sequence $\{w_k\}_{k\in\mathbb{N}}$ has tail index $\alpha>0$, regularly varying as in Definition~\ref{def_reg_var}, that is, $\mathbb{P}\!\left( \|w_k\| > r \right) \sim L(r) r^{-\alpha}, r \to \infty$.
The disturbance is thus allowed to be unbounded and heavy-tailed, so no bounded set $\mathcal{W}$ exists and the tube construction of Section~\ref{sec_tube_smpc} does not apply. We therefore impose the chance constraint~\eqref{eq_chance_const} under the following assumptions.
\begin{assumption}\label{assump1}
    (i) the pair $(A,B)$ in~\eqref{eq_sys} is controllable; (ii) the sets $\mathcal{X}$ and $\mathcal{U}$ contain the origin as an interior point; (iii) $\mathcal{U}$ is compact and convex.
\end{assumption}
\begin{assumption}\label{assump2}
    The disturbance sequence $\{w_k\}_{k \in \mathbb{N}}$ in~\eqref{eq_sys} is i.i.d. and regularly varying with tail index $\alpha > 0$.
\end{assumption}
 
Under Assumption~\ref{assump2}, the proposed approach targets the heavy-tailed regime, offering a deliberate complement to classical light-tailed SMPC~\citep{Mesbah2016_overview, FarinaGiulioniScattolini2016}. In practice, the tail-index estimator itself serves as a gatekeeper between the two regimes. It converges to a finite value on a genuinely heavy-tailed sample, and a diverging or unstable estimate on a light-tailed one is itself the signal to fall back to classical SMPC.

Throughout the paper, for simplicity we assume that $y$ is single-dimensional, so the constraint we aim to satisfy in~\eqref{eq_chance_const} is
\begin{equation}\label{eq_scalar_chance_const}
\mathbb{P}(c^\top x_k \le y_{\max}) \ge 1-\epsilon, \qquad k \in \mathbb{N},
\end{equation}
where $c^\top$ and $y_{\max}$ correspond to one row of $H$ and $h$ in~\eqref{eq_chance_const}, and $\epsilon$ is the corresponding entry of $\beta$. The remaining rows of~\eqref{eq_chance_const} are treated by the same analysis, applied separately along each row.
 \section{Rare-Event Constrained SMPC}\label{sec_rare_event_SMPC}
This section builds the EVT tightening of the chance constraint~\eqref{eq_scalar_chance_const}. Lemmas~\ref{lemma_state_tail_invariant} and~\ref{lemma_error_tail} collect known results on regularly varying sequences, stated in the form needed here.

\subsection{Tail invariance under linear feedback}
\begin{lemma}\label{lemma_state_tail_invariant}
Let Assumption~\ref{assump2} hold and let $u_k = -K x_k$ with $\rho(A_K)<1$, where $\rho(\cdot)$ denotes the spectral radius and $A_K = A-BK$. Extending the i.i.d.\ sequence $\{w_k\}$ to all integer times $k \in \mathbb{Z}$, the closed loop of~\eqref{eq_sys} admits the stationary solution $x_k = \sum_{j=0}^\infty A_K^j w_{k-1-j}$, which is the almost-sure limit of the state at time $k$ when the system is started from any fixed initial state at time $k-m$ and $m \to \infty$. This stationary state is regularly varying with the tail index $\alpha$ of the disturbance, namely
\begin{equation}\label{eq_020301}
\mathbb{P}\!\left( \|x_k\| > r \right)
\sim C_K r^{-\alpha},
\qquad r\to\infty,
\end{equation}
where the constant $C_K>0$ depends on the closed-loop dynamics.
\end{lemma}
\begin{proof}
Since $\rho(A_K)<1$, $\|A_K^j\|$ decays geometrically in $j$, so the series converges almost surely and defines the stationary solution. The coefficients $A_K^j$ are summable and $\{w_k\}$ is i.i.d.\ and regularly varying with index $\alpha$, so the tail-preservation theorem for linear processes in~\citep{hult2008tail} gives that $x_k$ is regularly varying with the same index $\alpha$.
\end{proof}
Lemma~\ref{lemma_state_tail_invariant} shows that the feedback gain $K$ changes only the scale $C_K$ of the state tail and never its index $\alpha$. No stabilizing feedback can therefore make the closed-loop state light-tailed, and rare excursions have to be handled at the level of the constraints. The tube-based scheme below does this by tightening the constraints on a nominal trajectory by an amount determined by the tail in~\eqref{eq_020301}.

\subsection{Tube decomposition}
Following tube MPC~\citep{TubeMPC_05}, given the current state $x_t$ and a prediction horizon $N$, the predicted input is split as $u_{i|t} = v_{i|t} - K e_{i|t}$, where $v_{i|t}$ is the nominal input and $e_{i|t} \triangleq \hat x_{i|t} - \bar x_{i|t}$ is the error between the predicted state $\hat x_{i|t}$ and the nominal state $\bar x_{i|t}$, with $\hat x_{0|t} = \bar x_{0|t} = x_t$. The nominal state evolves as $\bar x_{i+1|t} = A\bar x_{i|t} + Bv_{i|t}$, and the error evolves as
\begin{align}\label{eq_error_dyn}
e_{i+1|t} &= A_K e_{i|t} + w_{i|t}, \qquad e_{0|t} = 0,
\end{align}
where $w_{i|t}$ denotes the disturbance at step $t+i$. The predicted output therefore separates into a deterministic nominal part and a stochastic error part, $y_{i|t} = c^\top \bar x_{i|t} + Z_i$ with $Z_i \triangleq c^\top e_{i|t}$, and the chance constraint~\eqref{eq_scalar_chance_const} along the prediction requires $\mathbb{P}(y_{i|t}\le y_{\max}) \ge 1-\epsilon$ for $i=1,\dots,N$.

\subsection{EVT-based tightening of the chance constraint}
Classical tube MPC bounds $e_{i|t}$ by the reachable sets $\mathcal{R}^e_i$ of Section~\ref{sec_tube_smpc}, which do not exist under Assumption~\ref{assump2}. We instead characterize the tail of $Z_i$ directly. Solving~\eqref{eq_error_dyn} gives
\begin{equation}
Z_i = \sum_{j=0}^{i-1} c^\top A_K^j\, w_{i-1-j|t},
\end{equation}
a sum of $i$ independent terms. Each term is a fixed linear projection of $w$ and is therefore regularly varying with the same index $\alpha$ and a direction-dependent scale $\kappa_j \ge 0$, defined by
\begin{equation}\label{eq_kappa_def}
\mathbb{P}\big(c^\top A_K^j w > z\big) \sim \kappa_j\, z^{-\alpha}, \qquad z \to \infty,
\end{equation}
for $j = 0,\dots,i-1$. That is, $\kappa_j$ is the tail scale of the disturbance projected onto the direction $(A_K^{j})^{\top}c$, and it decays geometrically in $j$ together with $\|A_K^j\|^\alpha$.

\begin{lemma}\label{lemma_error_tail}
Let $Z_i = \sum_{j=0}^{i-1} c^\top A_K^j w_{i-1-j|t}$ with the tail scale coefficients $\kappa_j\ge0$, $j=0,\dots,i-1$, defined in~\eqref{eq_kappa_def}. Then
\begin{equation}\label{eq_error_tail}
\mathbb{P}(Z_i > z) \sim S_i(K)\, z^{-\alpha}, \qquad z \to \infty,
\end{equation}
where $S_i(K) \triangleq \sum_{j=0}^{i-1} \kappa_j$.
\end{lemma}
\begin{proof}
The summands of $Z_i$ are independent and regularly varying with the common index $\alpha$. For such sums, the tail of the sum is asymptotically equal to the sum of the individual tails, which is the one-big-jump principle~\citep[Ch.\ 2]{Resnick_2007}. Hence $\mathbb{P}(Z_i>z) \sim \sum_{j=0}^{i-1}\mathbb{P}\big(c^\top A_K^j w_{i-1-j|t}>z\big) \sim S_i(K)\,z^{-\alpha}$ as $z\to\infty$.
\end{proof}
$S_i(K)$ is a controller-dependent tail scale factor, computable via the recursion $S_{i+1}(K) = S_i(K) + \kappa_i$ with $S_0(K) = 0$. Since $\kappa_j \to 0$ geometrically, $S_i(K)$ converges to a finite limit $S_\infty(K)$ as $i \to \infty$. The gain $K$ may thus be chosen, subject to $\rho(A_K)<1$, to minimize $S_\infty(K)$ and reduce the conservatism of the resulting constraint tightening.

By~\eqref{eq_error_tail}, the $(1-\epsilon)$-quantile of $Z_i$ satisfies
\begin{equation}\label{eq_quantile_asymp}
q_{e,i}(1-\epsilon) \triangleq \inf\big\{ z \in \mathbb{R} : \mathbb{P}(Z_i \le z) \ge 1-\epsilon \big\} \sim \left( \frac{S_i(K)}{\epsilon} \right)^{1/\alpha}, \qquad \epsilon \to 0,
\end{equation}
so the rare-event threshold scales as $\epsilon^{-1/\alpha}$. Substituting~\eqref{eq_quantile_asymp} into the chance constraint yields the tightened, deterministic constraint on the nominal trajectory
\begin{equation}\label{eq_tightened_constraint}
c^\top \bar x_{i|t} \le y_{\max} - q_{e,i}(1-\epsilon), \qquad i = 1,\dots,N.
\end{equation}
Since $S_i(K)$ and $\alpha$ are not known a priori, we need to compute their estimates $\hat{S}_i(K)$ and $\hat{\alpha}$ from data and hence use an estimated quantile $\hat{q}_{e,i}(1-\epsilon)$. The estimation method is described as follows.

For any regularly varying $w_k$, the exceedances $c^\top w_k - \lambda$ of the projected disturbance $c^\top w_k$ above a high threshold $\lambda$ are well approximated, once $\lambda$ is large enough, by a Generalized Pareto distribution with shape parameter $\xi$ and scale parameter $\sigma_\lambda$, whose cumulative distribution function is
\begin{equation}
G_{\xi,\sigma_\lambda}(x) = 1 - \left(1 + \xi \frac{x}{\sigma_\lambda}\right)^{-1/\xi}, \quad x \ge 0,
\end{equation}
for $\xi \neq 0$~\citep{Coles2001}.
The shape parameter $\xi$ relates directly to the tail index through $\xi = 1/\alpha$. Hence, the maximum-likelihood fitting of $\xi$ and $\sigma_\lambda$ to the observed exceedances yields $\hat\alpha = 1/\hat\xi$. Combining $\hat\xi$, $\hat\sigma_\lambda$, and the empirical fraction of samples exceeding $\lambda$ through the standard peaks-over-threshold tail approximation then gives the base scale
$\hat\kappa_0 = \frac{n_\lambda}{n}\left(\frac{\hat\sigma_\lambda}{\hat\xi}\right)^{1/\hat\xi}$,
where $n_\lambda$ is the number of exceedances of $\lambda$ out of $n$ total samples. This threshold-exceedance estimator is a standard, consistent estimator of the tail index, that is, as the threshold $\lambda$ grows so that the exceedance count also grows while remaining a vanishing fraction of the sample, $\hat\alpha$ converges to $\alpha$ in probability~\citep{Coles2001}. The remaining coefficients $\hat\kappa_j$, $j\ge1$, then follow from $\hat\kappa_0$ and the already-known matrices $A_K^j$, without requiring additional data.

The terminal weight $P \succ 0$ in the MPC formulation is chosen as the solution to the discrete Lyapunov equation $A_K^\top P A_K - P = -Q$, so that the terminal cost $\bar x_{N|t}^\top P \bar x_{N|t}$ is a valid Lyapunov function for the nominal closed-loop dynamics, following the standard terminal-cost construction for nominal stability of MPC~\citep{MayneRawlingsRaoScokaert2000}. The resulting EVT-enhanced linear SMPC problem, with the feedback gain $K$ fixed and~\eqref{eq_tightened_constraint} imposed in place of~\eqref{eq_chance_const}, is
\be
\begin{aligned}
\min_{v_{0|t},\dots,v_{N-1|t}}
& \sum_{i=0}^{N-1} \bar x_{i|t}^\top Q \bar x_{i|t} + v_{i|t}^\top R v_{i|t} + \bar x_{N|t}^\top P \bar x_{N|t} \\
\text{s.t.}\quad
& \bar x_{0|t} = x_t, \\
& \bar x_{i+1|t} = A\bar x_{i|t} + Bv_{i|t}, \\
& v_{i|t} \in \mathcal{U}, \quad i = 0,\dots,N-1, \\
& c^\top \bar x_{i|t} \le y_{\max} - \hat q_{e,i}(1-\epsilon), \quad i = 1,\dots,N.
\end{aligned}
\label{eq_tightened_smpc}
\ee

\begin{remark}
Constraint~\eqref{eq_tightened_constraint} converts the stochastic constraint~\eqref{eq_chance_const} into a deterministic linear inequality on the nominal trajectory, with all tail characterization computed offline through $\hat q_{e,i}(1-\epsilon)$. The resulting EVT-SMPC problem thus retains the computational structure of nominal linear MPC.
\end{remark}

\begin{remark} 
The results of this section rely on the linear closed-loop map $x_{k+1}=A_Kx_k+w_k$, through which the disturbance tail propagates explicitly as in Lemma~\ref{lemma_state_tail_invariant}. For nonlinear dynamics $x_{k+1}=f(x_k,u_k)+w_k$, the formulation can be applied by linearizing the dynamics and the constraint along the nominal trajectory at every time step, in the spirit of tube-based nonlinear MPC, and by estimating the quantiles of the resulting time-varying error dynamics from Monte Carlo rollouts. Section~\ref{sec_sim_unicycle} details this procedure for a unicycle robot. The results of Sections~\ref{sec_rare_event_SMPC} and~\ref{sec_clustering} then apply to the linearized error dynamics, and certifying that the tail index and the clustering behavior persist under the nonlinear dynamics remains a direction for future work.
\end{remark}

\section{Clustering of Rare Events under Closed-Loop Dynamics}\label{sec_clustering}

A per-step chance constraint such as~\eqref{eq_tightened_constraint} bounds the exceedance probability at each prediction step in isolation. For a dynamical system, however, rare events are not independent across time, since a large disturbance persists through the closed-loop dynamics and continues to influence the state over several subsequent steps, so an exceedance at one step raises the likelihood of exceedance at the next. Because of this correlation, a marginal per-step probability of $\epsilon$ does not translate into the number of excursion episodes across the horizon. The rare-event constraint should therefore be posed at the level of a cluster of correlated exceedances, for which we use the notion of Leadbetter's extremal index.
\subsection{Rare-event cluster}
\begin{definition}[Extremal index~\citep{Leadbetter1983}]\label{def_extremal_index}
A stationary sequence $\{Z_k\}$ has extremal index $\theta \in (0,1]$ if, for every $\tau>0$, (i) there exists a threshold sequence $\mu_n(\tau)$ such that $n\,\mathbb{P}(Z_1>\mu_n(\tau)) \to \tau$ as $n \rightarrow \infty$ and (ii) the running maximum condition
\[
\mathbb{P}\Big(\max_{1\le k\le n} Z_k \le \mu_n(\tau)\Big) \to \exp(-\theta\tau) \quad \text{as } n\to\infty,
\]
is satisfied.
\end{definition}
\begin{figure}[htb]
    \centering
    \includegraphics[width=0.9\linewidth]{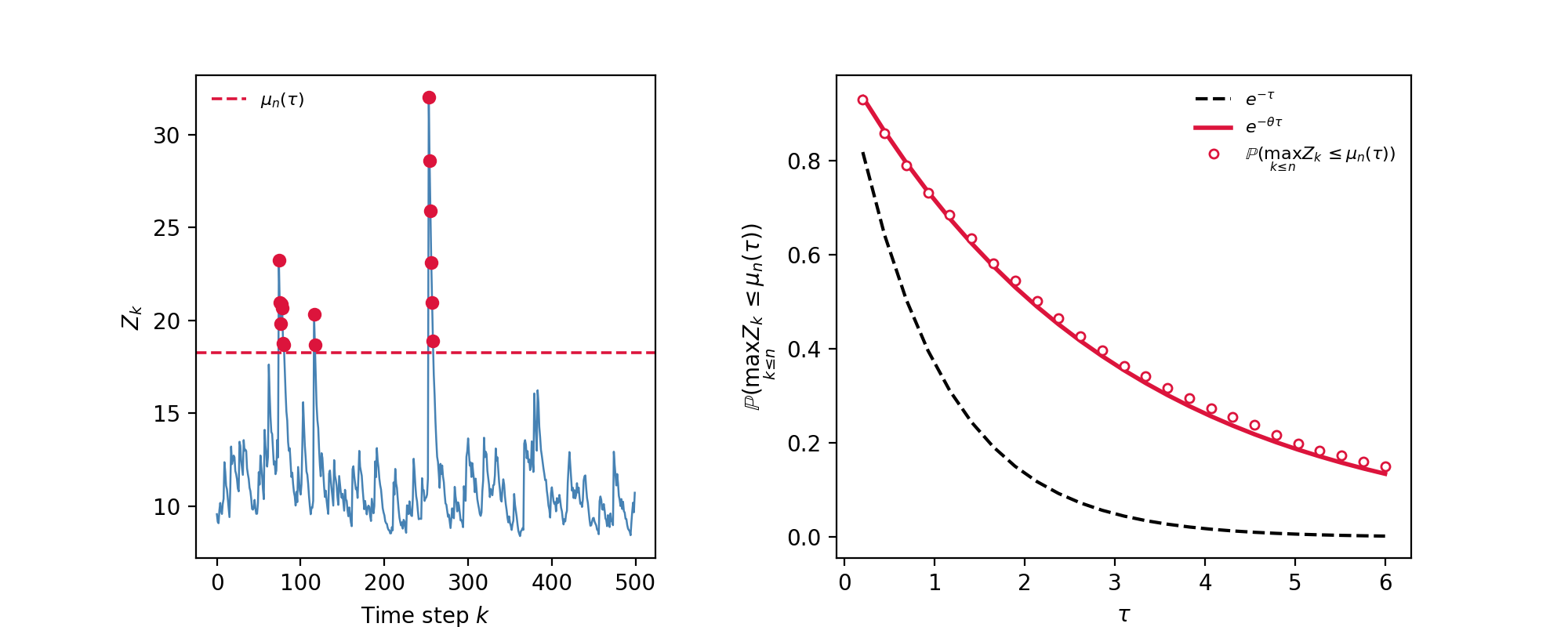}
    \caption{Extremal index illustration for process $Z_k=\rho Z_{k-1}+\zeta_k$, where $\zeta_k$ are i.i.d. from a Pareto distribution with tail index $\alpha = 2.5$. Left: clustering of exceedances above $\mu_n(\tau)$. Right: running maximum probability against $\tau$.}
    \label{fig_extremal_index_illustration}
\end{figure}
\FloatBarrier

Figure~\ref{fig_extremal_index_illustration} illustrates the extremal index. Exceedances above the threshold $\mu_n(\tau)$ arrive in a small number of multi-step clusters rather than scattering uniformly in time, since a large $\zeta_k$ persists through the dynamics and keeps several subsequent steps elevated. The running maximum probability $\mathbb{P}(\max_{k\le n}Z_k\le \mu_n(\tau))$ follows $e^{-\theta\tau}$ rather than the independent reference $e^{-\tau}$.
A value $\theta=1$ indicates that marginal exceedances occur essentially independently, whereas $\theta<1$ indicates clustering. In the latter case, conditional on one exceedance, subsequent steps are far more likely to also exceed the threshold, with expected cluster length $1/\theta$~\citep{Coles2001}. Theorem~\ref{thm_extremal_index} below gives the extremal index $\theta \in (0,1]$ of the closed-loop projection in closed form. It is strictly less than one whenever more than one coefficient of the closed-loop response contributes to the tail with the same sign, which is the generic case for a dynamical system, so the marginal exceedance events are then not independent across time.

A consequence follows for the design of the rare-event constraint. A cluster of correlated exceedances forms a single rare-event episode, and the design objective should therefore bound the probability of such an episode over the horizon rather than the marginal per-step exceedance probability enforced by~\eqref{eq_tightened_constraint}. The required correction follows from the running-maximum condition of Definition~\ref{def_extremal_index}, and it is at this step that the analysis goes beyond a direct application of extreme value theory. A per-step EVT quantile treats each prediction step in isolation, whereas the closed-loop dynamics make the exceedances at consecutive steps dependent, and the extremal index $\theta$ is the parameter that accounts for this dependence. Suppose that a common per-step level $\epsilon$ is enforced at every step of a horizon of length $N$, and let $\mu$ denote the corresponding $(1-\epsilon)$-quantile of $Z_1$. If the exceedances were independent, the probability of no exceedance over the horizon would be $(1-\epsilon)^N \approx \exp(-N\epsilon)$. Under clustering, Definition~\ref{def_extremal_index} with $\tau = N\epsilon$ gives instead
\[
\mathbb{P}\Big(\max_{1\le k\le N} Z_k \le \mu\Big) \approx \exp(-\theta N\epsilon).
\]
That is, the $N\epsilon$ exceedances expected over the horizon arrive in about $\theta N \epsilon$ independent clusters of mean length $1/\theta$, and the probability of at least one excursion episode over the horizon is approximately $1-\exp(-\theta N\epsilon)$. This probability is much larger than the target $\epsilon$ unless $\theta N$ is small. Setting it equal to $\epsilon$ and solving for the per-step level gives
\begin{equation}\label{eq_theta_corrected_epsilon}
\epsilon_\star \triangleq \frac{-\ln(1-\epsilon)}{\theta N}.
\end{equation}
In~\eqref{eq_theta_corrected_epsilon}, the dependence enters only through the product $\theta N$, which plays the role of the effective number of independent steps in the horizon. Using $\epsilon_\star$ in place of $\epsilon$ in~\eqref{eq_tightened_constraint} gives the clustering-corrected constraint
\begin{equation}\label{eq_theta_corrected_constraint}
c^\top \bar x_{i|t} \le y_{\max} - \hat q_{e,i}(1-\epsilon_\star), \qquad i = 1,\dots,N,
\end{equation}
which differs from~\eqref{eq_tightened_constraint} only in the per-step probability level, and is computed by the same offline procedure, so it requires no change to the optimization structure of Section~\ref{sec_rare_event_SMPC}. Whenever $\theta N > -\ln(1-\epsilon)/\epsilon$, it follows that $\epsilon_\star < \epsilon$, so the corrected constraint is tighter than the naive per-step constraint. Since $-\ln(1-\epsilon)/\epsilon$ is close to one for small $\epsilon$, this holds as soon as $\theta N$ exceeds one by a small margin. Compared with the level $-\ln(1-\epsilon)/N$, which spreads the risk over $N$ steps as if every exceedance were a separate episode, $\epsilon_\star$ is larger by the factor $1/\theta$, so accounting for clustering also removes the conservatism of that choice. Theorem~\ref{thm_theta_corrected} in Section~\ref{sec_theta_estimation} makes this derivation precise.

\subsection{Extremal index of the closed-loop projection}
For the stationary closed-loop state of Lemma~\ref{lemma_state_tail_invariant}, the projected output is $Z_k \triangleq c^\top x_k = \sum_{j=0}^\infty c^\top A_K^j w_{k-1-j}$. It has the same distribution as the limit of the projected tube error $Z_i$ of Section~\ref{sec_rare_event_SMPC} as $i \to \infty$, since $Z_i$ is the same weighted sum of i.i.d.\ disturbances truncated after $i$ terms. We use the time index $k$ for the stationary sequence and the prediction index $i$ for the tube error. We further assume that the tail-relevant randomness of the disturbance is captured by a single fixed direction $b$, as follows.  
\begin{assumption}\label{assump_dominant_direction}
The disturbance's tail-relevant randomness is effectively one-dimensional: there is a fixed direction $b \in \mathbb{R}^{n_x}$, $\|b\|=1$, such that $w_k$ admits the tail-equivalent decomposition $w_k = \zeta_k b + \eta_k$, where $\zeta_k \triangleq b^\top w_k$ is a scalar regularly varying random variable with index $\alpha$ and tail-balance parameters $(p,q)$ as in Definition~\ref{def_tail_balance}, and $\eta_k$ is negligible in the tail relative to $\zeta_k b$, i.e., $\mathbb{P}(\|\eta_k\|>z) = o\big(\mathbb{P}(|\zeta_k|>z)\big)$ as $z\to\infty$.
\end{assumption}
\begin{remark} 
Assumption~\ref{assump_dominant_direction} is the scalar specialization of the general notion of multivariate regular variation. When several disturbances contribute comparably to the tail, $w_k$ needs to be modeled as supported over more than one direction. The tail-scale decomposition~\eqref{eq_kappa_def} and the tail-invariance result of Lemma~\ref{lemma_state_tail_invariant} extend directly to this setting. However, the extremal index derivation does not extend as directly. Once the disturbance's extremal mass is spread over a continuum of directions, no single parameter plays the role of $\theta$ in Definition~\ref{def_extremal_index}. We leave the multi-dimensional heavy-tailed disturbance case to future work.
\end{remark}

Substituting $w_k$ into $Z_k = \sum_{j=0}^\infty c^\top A_K^j w_{k-1-j}$ and retaining the tail-dominant term yields
$Z_k \sim \sum_{j=0}^\infty g_j\, \zeta_{k-1-j}, \qquad g_j \triangleq c^\top A_K^j b$,
that is, $Z_k$ and $\sum_{j=0}^\infty g_j\, \zeta_{k-1-j}$ are tail-equivalent, namely $\lim_{z \to \infty} \frac{\mathbb{P}(Z_k > z)}{\mathbb{P}\big(\sum_j g_j \zeta_{k-1-j} > z\big)} = 1$.

Since $\{w_k\}$ are i.i.d. from Assumption~\ref{assump2} and $b$ is fixed, $\{\zeta_k\} = \{b^\top w_k\}$ are also i.i.d. Assumption~\ref{assump_dominant_direction} additionally specifies its regular variation and tail structure. Deriving the extremal index for $Z_k$ requires two ingredients: the marginal tail of $Z_k$ and a running-maximum limit for the linear process $Z_k$, which we establish in turn.

\begin{lemma}\label{lemma_conv_closure}
Let $\zeta_1,\dots,\zeta_m$ be i.i.d., regularly varying with index $\alpha$ and tail-balance parameters $(p,q)$, and let $a_1,\dots,a_m \in \mathbb{R}$. Then
\begin{equation}\label{eq_finite_sum_tail}
\mathbb{P}\Big(\sum_{j=1}^m a_j \zeta_j > z\Big) \sim \Big[p\sum_{j=1}^m (a_j^+)^\alpha + q\sum_{j=1}^m (a_j^-)^\alpha\Big] \mathbb{P}(|\zeta_1|>z),
\end{equation}
as $z \to \infty$.
Here, $a_j^+ = a_j \vee 0$ and $a_j^- = -(a_j \wedge 0)$.
\end{lemma}

\begin{proof}
We prove by induction on $m$. We start with the case when $m=1$. Suppose $a_1 > 0$. Then
$\mathbb{P}(a_1\zeta_1 > x) = \mathbb{P}\Big(\zeta_1 > \frac{x}{a_1}\Big)$.
By regular variation, $\mathbb{P}(|\zeta_1|>t) = t^{-\alpha}L(t)$ for a slowly varying function $L$, and by the tail-balance property, $\mathbb{P}(\zeta_1>t) \sim p\,\mathbb{P}(|\zeta_1|>t)$ as $t\to\infty$. Setting $t = x/a_1$,
\[
\mathbb{P}\Big(\zeta_1 > \frac{x}{a_1}\Big) \sim p \left(\frac{x}{a_1}\right)^{-\alpha} L\!\left(\frac{x}{a_1}\right) = p\, a_1^\alpha\, x^{-\alpha}\, L\!\left(\frac{x}{a_1}\right).
\]
Since $L$ is slowly varying, $L(x/a_1)/L(x) \to 1$ as $x\to\infty$, which yields that $\mathbb{P}(a_1\zeta_1>x)$ is
\[
p\,a_1^\alpha\, x^{-\alpha}L(x) = p\,a_1^\alpha\,\mathbb{P}(|\zeta_1|>x) = p\,(a_1^+)^\alpha\,\mathbb{P}(|\zeta_1|>x).
\]
For $a_1<0$, dividing by $a_1$ reverses the inequality: $a_1\zeta_1>x \iff \zeta_1 < x/a_1 = -x/|a_1|$. By the tail-balance property for the lower tail, $\mathbb{P}(\zeta_1 \le -t) \sim q\,\mathbb{P}(|\zeta_1|>t)$, and the same slowly-varying argument gives
\[
\mathbb{P}(a_1\zeta_1>x) \sim q\,|a_1|^\alpha\,\mathbb{P}(|\zeta_1|>x) = q\,(a_1^-)^\alpha\,\mathbb{P}(|\zeta_1|>x).
\]
Together, these establish~\eqref{eq_finite_sum_tail} for $m=1$. 

Now, we suppose~\eqref{eq_finite_sum_tail} holds for $m-1$ and prove that it holds for $m$ consequently. By the convolution closure property of regularly varying tails~\citep[Ch.\ 2]{Resnick_2007}, for independent regularly varying random variables $U$ and $V$ with a common index $\alpha$,
\[
\mathbb{P}(U+V>x) \sim \mathbb{P}(U>x) + \mathbb{P}(V>x), \qquad x\to\infty.
\]
Set $U = \sum_{j=1}^{m-1} a_j\zeta_j$, which is regularly varying by the inductive hypothesis, and $V = a_m\zeta_m$, also regularly varying by the base case; $U$ and $V$ are independent since they involve disjoint subsets of the i.i.d.\ sequence $\{\zeta_j\}$. Then $\mathbb{P}\Big(\sum_{j=1}^{m} a_j\zeta_j > x\Big)$ is $\Big[p\sum_{j=1}^{m-1}(a_j^+)^\alpha + q\sum_{j=1}^{m-1}(a_j^-)^\alpha\Big]\mathbb{P}(|\zeta_1|>x) + \Big[p\,(a_m^+)^\alpha + q\,(a_m^-)^\alpha\Big]\mathbb{P}(|\zeta_1|>x) = \Big[p\sum_{j=1}^{m}(a_j^+)^\alpha + q\sum_{j=1}^{m}(a_j^-)^\alpha\Big]\mathbb{P}(|\zeta_1|>x)$,
which establishes~\eqref{eq_finite_sum_tail} for $m$, completing the induction.
\end{proof}

\begin{lemma}\label{prop_marginal_tail}
Under Assumption~\ref{assump_dominant_direction}, from Lemma~\ref{lemma_conv_closure}, it follows that
\begin{equation*}
\mathbb{P}(Z_k > z) \sim \Big[p\sum_{j=0}^\infty (g_j^+)^\alpha + q\sum_{j=0}^\infty (g_j^-)^\alpha\Big] \mathbb{P}(|\zeta_1|>z), z \to \infty.
\end{equation*}
\end{lemma}
\begin{proof}
By Lemma~\ref{lemma_conv_closure}, for each fixed $m$ the truncated sum $Z_k^{(m)} \triangleq \sum_{j=0}^m g_j\zeta_{k-1-j}$ satisfies
\[
\mathbb{P}(Z_k^{(m)} > z) \sim \Big[p\sum_{j=0}^m (g_j^+)^\alpha + q\sum_{j=0}^m (g_j^-)^\alpha\Big] \mathbb{P}(|\zeta_1|>z), z\to\infty.
\]
Since $\{g_j\}$ is deterministic (as $A_K$, $c$, $b$ are fixed) and $|g_j| \le \|c\|\bar C_K\rho_K^j$ decays geometrically, where $\bar C_K>0$ and $\rho_K \in [0,1)$ are constants with $\|A_K^j\| \le \bar C_K\rho_K^j$, which exist since $\rho(A_K)<1$, the summability conditions of~\citep[Eqs.\ (3.2)--(3.3)]{hult2008tail}, that is, $\sum_j|g_j|^{\alpha-\gamma}<\infty$ for some $\gamma>0$ if $\alpha\le2$, or $\sum_j g_j^2<\infty$ if $\alpha>2$, hold. The extension of the finite truncation to the full infinite series $Z_k = \sum_{j=0}^\infty g_j\zeta_{k-1-j}$ then follows from~\citep[Corollary\ 3.1, Remark\ 3.3]{hult2008tail}.  \end{proof}

With the marginal tail of $Z_k$ established, we now turn to its running-maximum behavior. We first state a general running-maximum limit for linear processes of regularly varying variables, which we then specialize to derive the extremal index for $Z_k$.

\begin{lemma}[Th. 3.1~\citep{davis1985limit}]\label{lemma_moving_avg_max}
Let $\{\zeta_k\}$ be i.i.d., regularly varying with index $\alpha$ and tail-balance parameters $(p,q)$, and let $\{g_j\}_{j\ge0}$ be a deterministic sequence with $\sum_{j=0}^\infty |g_j|^\delta < \infty$ for some $\delta<\alpha$, $\delta\le1$. Let $r_n$ denote the normalizing sequence such that $n\,\mathbb{P}(|\zeta_1|>r_n z)\to z^{-\alpha}$ for all $z>0$. Then as $n\rightarrow \infty$, $Z_k \triangleq \sum_{j=0}^\infty g_j\zeta_{k-1-j}$ satisfies, for every fixed $z>0$,
\begin{equation}\label{eq_moving_avg_max_limit}
\mathbb{P}\big(r_n^{-1}\max_{k\le n} Z_k \le z\big) \to \exp\big(-(g_+^\alpha p + g_-^\alpha q)z^{-\alpha}\big),
\end{equation}
with $g_+ \triangleq \sup_j(g_j\vee 0) = \sup_j g_j^+$ and $g_- \triangleq \sup_j(-g_j\vee 0) = \sup_j g_j^-$.
\end{lemma}

\begin{Thm}\label{thm_extremal_index}
Suppose $\{\zeta_k\}$ is i.i.d.\ and regularly varying with index $\alpha$ and tail-balance parameters $(p,q)$ as in Assumption~\ref{assump_dominant_direction}, and $\|A_K^j\| \le \bar C_K \rho_K^j$ for some constants $\bar C_K>0$ and $\rho_K \in [0,1)$, which exist since $\rho(A_K)<1$ as in Lemma~\ref{lemma_state_tail_invariant}, so that $\sum_{j=0}^\infty |g_j|^\delta < \infty$ for some $\delta<\alpha$, $\delta \le 1$. Then $Z_k$ is regularly varying with tail index $\alpha$, and its running maximum satisfies~\eqref{eq_moving_avg_max_limit}.
Consequently, $\{Z_k\}$ admits the extremal index
\begin{equation}\label{eq_extremal_index_formula}
\theta = \frac{g_+^\alpha p + g_-^\alpha q}{\displaystyle p\sum_{j=0}^\infty (g_j^+)^\alpha + q\sum_{j=0}^\infty (g_j^-)^\alpha} \in (0,1].
\end{equation}
\end{Thm}
\begin{proof}
Since $\{g_j\}$ is deterministic and decays geometrically, Lemma~\ref{lemma_moving_avg_max} applies to $Z_k = \sum_j g_j\zeta_{k-1-j}$, giving, for every fixed $z>0$,
\begin{equation}\label{eq_running_max_cited}
\mathbb{P}\big(r_n^{-1}\max_{k\le n} Z_k \le z\big) \to \exp(-\Lambda z^{-\alpha}), \Lambda \triangleq g_+^\alpha p + g_-^\alpha q,
\end{equation}
where $r_n$ is the normalizing sequence such that $n\,\mathbb{P}(|\zeta_1|>r_n z)\to z^{-\alpha}$ for all $z>0$. By Lemma~\ref{prop_marginal_tail}, $\mathbb{P}(Z_1>y)\sim D\,y^{-\alpha}L(y)$ as $y\to\infty$, with $D \triangleq p\sum_j(g_j^+)^\alpha+q\sum_j(g_j^-)^\alpha$ and $L$ the slowly varying function associated with $\zeta_1$.

Fix $\tau>0$ and set $\mu_n \triangleq r_n\,(D/\tau)^{1/\alpha} = r_n s$, that is, $\mu_n(\tau)$ with the $\tau$-dependence suppressed for brevity throughout the remainder of this proof. Then
\[
n\,\mathbb{P}(Z_1>\mu_n) \sim n D (r_n s)^{-\alpha} L(r_n s) = D s^{-\alpha} \left[n r_n^{-\alpha} L(r_n s)\right].
\]
Hence,
\[
n\,\mathbb{P}(Z_1>\mu_n) \sim   D\,(D/\tau)^{-1}\cdot\big[n\,r_n^{-\alpha}L(r_n(D/\tau)^{1/\alpha})\big] \to \tau,
\]
using that $L$ is slowly varying, so $L(r_n(D/\tau)^{1/\alpha})/L(r_n)\to1$ as $n\to\infty$. This verifies Leadbetter's condition $n\,\mathbb{P}(Z_1>\mu_n)\to\tau$ in Definition~\ref{def_extremal_index}.
Substituting $z = \mu_n/r_n = (D/\tau)^{1/\alpha}$ into~\eqref{eq_running_max_cited}, it yields $\mathbb{P}\big(\max_{k\le n}Z_k\le \mu_n\big)$ satisfies
\[
\mathbb{P}\big(r_n^{-1}\max_{k\le n}Z_k \le (D/\tau)^{1/\alpha}\big) \to \exp \left(-\Lambda \tau /D \right) = \exp(-\theta\tau),
\]
which is Leadbetter's defining condition with $\theta = \Lambda/D = (g_+^\alpha p+g_-^\alpha q)/D$.
\end{proof}

\begin{remark}
A Fréchet distribution with shape parameter $\alpha>0$ and scale parameter $\sigma>0$ has cumulative distribution function $F(x) = \exp\big(-(x/\sigma)^{-\alpha}\big)$ for $x>0$. The limit $\exp\big(-(g_+^\alpha p + g_-^\alpha q)\,x^{-\alpha}\big)$ in Theorem~\ref{thm_extremal_index} takes exactly this form, with shape parameter $\alpha$ and scale parameter $(g_+^\alpha p + g_-^\alpha q)^{1/\alpha}$. That is, the running maximum $r_n^{-1}\max_{k\le n}Z_k$ converges in distribution to a Fréchet distribution as $n\to\infty$.
\end{remark}

\subsection{Extremal Index Estimation}\label{sec_theta_estimation}
Theorem~\ref{thm_extremal_index} gives $\theta$ in closed form, derived from the model primitives $A_K$, $b$, $c$, $\alpha$, $p$, and $q$. In practical computation, we should instead use a statistical estimator of this quantity from a simulated trajectory. Such an estimator is called the intervals estimator~\citep{FerroSegers2003}. This estimator is not a contribution of this paper.  
We review it for the completeness of our proposed method.

Given a long trajectory of $Z_k$ and a threshold $\mu_n$, let $m$ denote the number of exceedances of $\mu_n$, let $s_1<\dots<s_m$ denote the exceedance times, and let $\Delta_i=s_{i+1}-s_i$, $1\le i \le m-1$, denote the inter-exceedance times. As $\mu_n$ approaches the upper endpoint of the marginal distribution of $Z_k$, the rescaled inter-exceedance times converge in distribution to a two-component mixture, a point mass at zero with probability $1-\theta$, corresponding to exceedances within the same cluster, and an exponential random variable with mean $1/\theta$ with probability $\theta$, corresponding to the waiting time between clusters. A sequence with $\theta=1$ therefore produces inter-exceedance times that are asymptotically memoryless, and clustering appears as an over-dispersion of the observed $\Delta_i$ relative to this baseline.

Matching the first two moments of the limiting mixture, $\mathbb{E}(\Delta_\theta)=1$ and $\mathbb{E}(\Delta_\theta^2)=2/\theta$, to their empirical counterparts gives the first branch of the estimator below. A second, bias-corrected branch instead matches moments of the inter-exceedance distribution at finite thresholds directly, which shrinks the smallest observed gaps toward zero and removes the first-order bias present in the first branch, but is only defined once $\max_i \Delta_i>2$. Both branches are truncated at one, since the finite-sample moment ratio can exceed the theoretical upper bound $\theta\le1$. The resulting estimator is
\begin{equation*}
\hat\theta = \begin{cases} \min\!\Big(1,\ \dfrac{2\big(\sum_i \Delta_i\big)^2}{(m-1)\sum_i \Delta_i^2}\Big), & \max_i \Delta_i \le 2,\\[10pt] \min\!\Big(1,\ \dfrac{2\big(\sum_i (\Delta_i-1)\big)^2}{(m-1)\sum_i (\Delta_i-1)(\Delta_i-2)}\Big), & \max_i \Delta_i > 2. \end{cases}
\end{equation*}
Ferro and Segers~\citep{FerroSegers2003} established that this estimator is \emph{consistent} as the threshold approaches the upper endpoint under standard mixing conditions on $\{Z_k\}$.

With the corrected target $\epsilon_\star$ from~\eqref{eq_theta_corrected_epsilon} in place of $\epsilon$, the tightened EVT-SMPC problem in Section~\ref{sec_rare_event_SMPC} becomes
\begin{align}
\nonumber
\min_{v_{0|t},\dots,v_{N-1|t}}
& \sum_{i=0}^{N-1} \bar x_{i|t}^\top Q \bar x_{i|t} + v_{i|t}^\top R v_{i|t} + \bar x_{N|t}^\top P \bar x_{N|t} \\
\text{s.t.} 
& \bar x_{0|t} = x_t, \label{eq_theta_corrected_smpc} \\
\nonumber
& \bar x_{i+1|t} = A\bar x_{i|t} + Bv_{i|t}, \\
\nonumber
& v_{i|t} \in \mathcal{U}, \quad i = 0,\dots,N-1, \\
\nonumber
& c^\top \bar x_{i|t} \le y_{\max} - \hat q_{e,i}(1-\epsilon_\star), \quad i = 1,\dots,N.
\end{align}

\begin{remark}
This problem retains exactly the objective, dynamics, and input constraints of the EVT-SMPC problem in Section~\ref{sec_rare_event_SMPC}, and differs only in the tightening level of the state constraint, so it requires no new solver or algorithmic structure, only a smaller $\epsilon_\star$ in the offline quantile computation already used there.
\end{remark}
The following result states in what sense the last constraint of this problem enforces an episode-level chance constraint.

\begin{Thm}\label{thm_theta_corrected}
Let the conditions of Theorem~\ref{thm_extremal_index} hold and let $Z_1$ have a continuous distribution. For $\epsilon \in (0,1)$ and a horizon $N$, let $\epsilon_\star$ be given by~\eqref{eq_theta_corrected_epsilon} and let $\mu_N$ be the $(1-\epsilon_\star)$-quantile of $Z_1$. Assume that the projected tube errors along the horizon have the same joint distribution as $N$ consecutive values of the stationary sequence $\{Z_k\}$, so that $q_{e,i}(1-\epsilon_\star) = \mu_N$ for $i=1,\dots,N$. Then, under the last constraint of~\eqref{eq_theta_corrected_smpc}, the probability of at least one constraint violation over the horizon is at most $\mathbb{P}\big(\max_{1\le k\le N} Z_k > \mu_N\big)$, with equality when the constraint is active at every step, and
\begin{equation}\label{eq_episode_limit}
\lim_{N\to\infty} \mathbb{P}\Big(\max_{1\le k\le N} Z_k > \mu_N\Big) = \epsilon.
\end{equation}
\end{Thm}
\begin{proof}
Since $c^\top \bar x_{i|t} \le y_{\max} - \mu_N$, a violation at step $i$ implies $Z_i > \mu_N$, with the converse holding when the constraint is active, which gives the bound. Let $\tau \triangleq -\ln(1-\epsilon)/\theta > 0$. Since $Z_1$ has a continuous distribution, $N\,\mathbb{P}(Z_1>\mu_N) = N\epsilon_\star = \tau$ for every $N$, and $\mu_N \to \infty$ as $N \to \infty$. Let $\mu_N' \triangleq r_N (D/\tau)^{1/\alpha}$ be the threshold used in the proof of Theorem~\ref{thm_extremal_index}, for which $N\,\mathbb{P}(Z_1>\mu_N') \to \tau$. By Lemma~\ref{prop_marginal_tail}, the tail of $Z_1$ is regularly varying with index $\alpha$, so $\mathbb{P}(Z_1>\mu_N)/\mathbb{P}(Z_1>\mu_N') \to 1$ implies $\mu_N/\mu_N' \to 1$, and hence $z_N \triangleq \mu_N/r_N \to (D/\tau)^{1/\alpha}$. By~\eqref{eq_running_max_cited}, $r_N^{-1}\max_{k\le N}Z_k$ converges in distribution to a Fréchet law with the continuous distribution function $G(z) = \exp(-\Lambda z^{-\alpha})$. Therefore
\begin{align*}
\mathbb{P}\Big(\max_{1\le k\le N} Z_k \le \mu_N\Big) &= \mathbb{P}\Big(r_N^{-1}\max_{1\le k\le N} Z_k \le z_N\Big) \\
&\to G\big((D/\tau)^{1/\alpha}\big) = \exp(-\theta\tau) = 1-\epsilon,
\end{align*}
which proves~\eqref{eq_episode_limit}.
\end{proof}
In the EVT-SMPC problem the predicted tube error starts from $e_{0|t}=0$, and its scale $S_i(K)$ approaches $S_\infty(K)$ geometrically in $i$, so the stationarity assumption of Theorem~\ref{thm_theta_corrected} holds approximately after a few prediction steps. For a finite horizon the episode probability differs from $\epsilon$, and Section~\ref{sec_sim_cstr} quantifies this difference.

\section{Simulation}\label{sec_simulation}
\subsection{Two-CSTR Process}\label{sec_sim_cstr}
We first consider the two-CSTR-in-series process of~\citep{YeTang2025AIChE}, linearized around its steady state and discretized in time,
\begin{equation*}
x_{k+1} = A x_k + B u_k + w_k,
\end{equation*}
with the system matrices $A$ and $B$ of~\citep{YeTang2025AIChE}. In the simulation, this linearized model is taken as the true dynamics of the process, so that the results of Sections~\ref{sec_rare_event_SMPC} and~\ref{sec_clustering} apply directly. The state $x_k = (C_{A1}, T_1, C_{A2}, T_2)$ collects the deviations of the reactant concentrations and temperatures in the two reactors, and the input $u_k$ collects the deviations of the inlet concentrations and jacket heat rates. The disturbance $w_k$ acts additively on all four states, with each component drawn i.i.d.\ from a Student-$t_\nu$ distribution with $\nu=3$ degrees of freedom. The safety requirement is to keep the temperature deviation $T_2$ of the second reactor below $T_{2, \max}=2$, while the controller tracks the setpoint $T_2=3$ beyond this limit. The limit and the setpoint are chosen in deviation variables so that the setpoint lies beyond the limit and the constraint is active. The MPC prediction horizon is $N=12$, and the feedback gain $K$ is the LQR gain with $Q=0.1I$ and $R=I$, for which $\rho(A_K)=0.51$.

We compare four controllers. The baseline is the same MPC without any tightening, which serves only as a naive reference. The Gaussian and EVT controllers solve the tightened tube SMPC problem~\eqref{eq_tightened_smpc} and differ only in how they tighten the temperature constraint, using a light-tailed quantile estimate and a generalized Pareto quantile estimate, respectively. As in Section~\ref{sec_sim_unicycle}, the tightened constraint is softened by a heavily penalized slack variable. The fourth controller is the $\theta$-corrected EVT controller of Section~\ref{sec_clustering}. Since the system is linear and time-invariant, all tightening levels are computed once offline from $2\times10^4$ Monte Carlo samples of the tube error.

Fig.~\ref{fig_cstr_closed_loop}(a) shows $T_2$ over $81$ time steps of a closed-loop run with $\epsilon=10^{-3}$, centered at the step where the disturbance on $T_2$ is largest. All controllers are driven by the same disturbance realization. Before this step, each controller keeps $T_2$ below $T_{2, \max}$ by a margin set by its tightening. The same large disturbance then raises $T_2$ by the same amount for all controllers, so that it violates the limit under the Gaussian and EVT tightenings but not under the $\theta$-corrected one, which operates furthest from the limit. Fig.~\ref{fig_cstr_closed_loop}(b) reports the realized per-step violation rates over $4\times10^4$ time steps.
\begin{figure}[htb]
    \centering
    \includegraphics[width=0.9\linewidth]{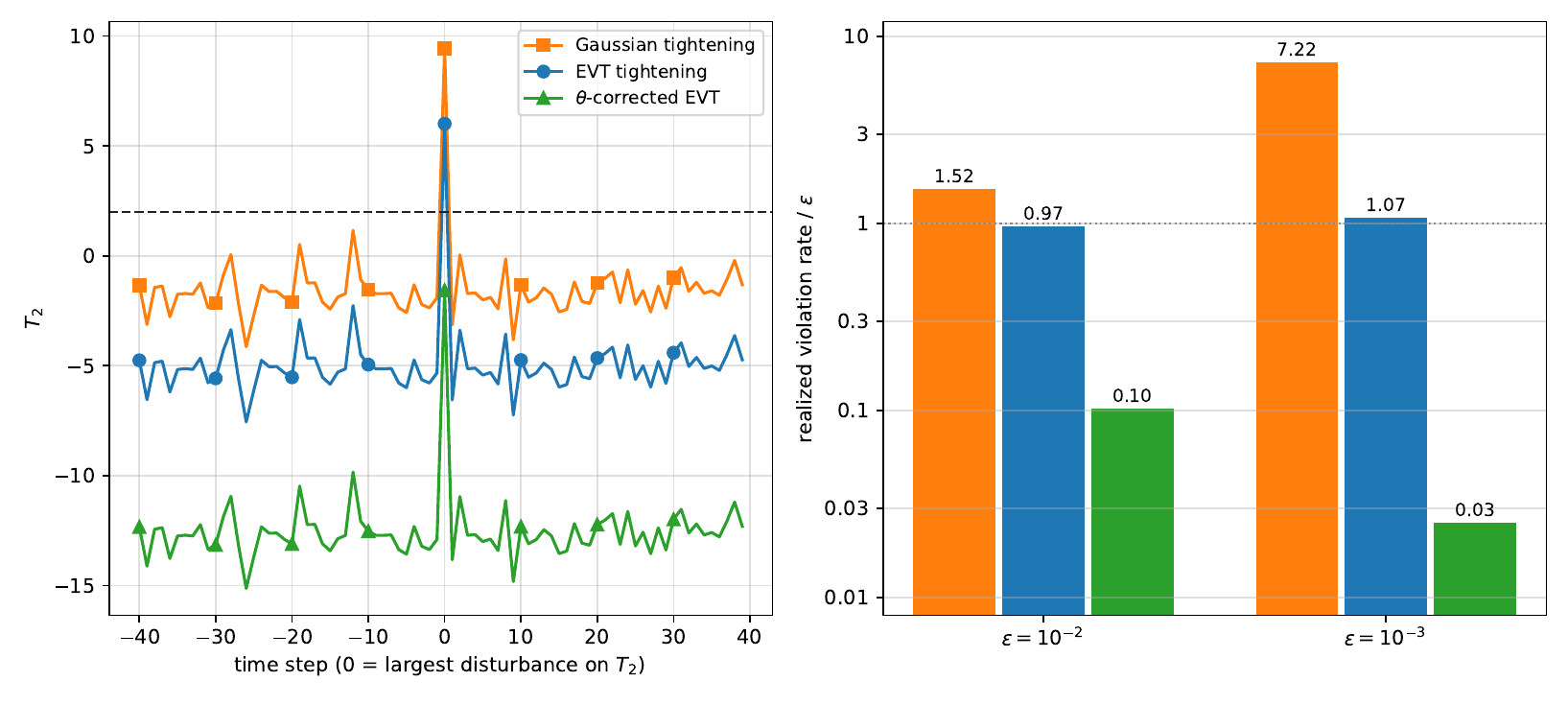}
    \caption{Closed-loop results for the two-CSTR process. (a) $T_2$ around its largest disturbance in a closed-loop run with $\epsilon=10^{-3}$, where step $0$ is the step of this disturbance and the dashed line marks $T_{2, \max}$. (b) Realized per-step violation rate relative to $\epsilon$ over $4\times10^4$ time steps.}
    \label{fig_cstr_closed_loop}
\end{figure}
\FloatBarrier

For $\epsilon=10^{-3}$, the baseline violates the limit at $48\%$ of the steps. The Gaussian controller violates it at $0.72\%$ of the steps, seven times the target, since the Gaussian quantile underestimates the heavy-tailed disturbance, whereas the EVT controller attains $0.11\%$, close to the target, with a larger mean margin of $7.06$ against $3.64$. At $\epsilon=10^{-2}$, the two violation rates are $1.52$ and $0.97$ times the target.

The extremal index estimated from the closed loop is $\hat\theta=0.92$, which for $N=12$ corrects the per-step target from $\epsilon=10^{-3}$ to $\epsilon_\star=9.11\times10^{-5}$. Fig.~\ref{fig_cstr_prediction} shows the effect of this correction. The left panel reports the probability that the predicted tube error exceeds its tightening at least once over the horizon, evaluated on $2\times10^6$ unseen samples, and the right panel reports the minimum safety margin $T_{2, \max}-T_2$ over consecutive windows of $120$ time steps of the closed-loop run. The $\theta$-corrected controller reduces the violation probability over the horizon from $0.0120$ for the naive EVT controller to $0.0015$, much closer to the target, and shifts the margin distribution upward at the cost of a larger distance from the setpoint.
\begin{figure}[htb]
    \centering
    \includegraphics[width=0.9\linewidth]{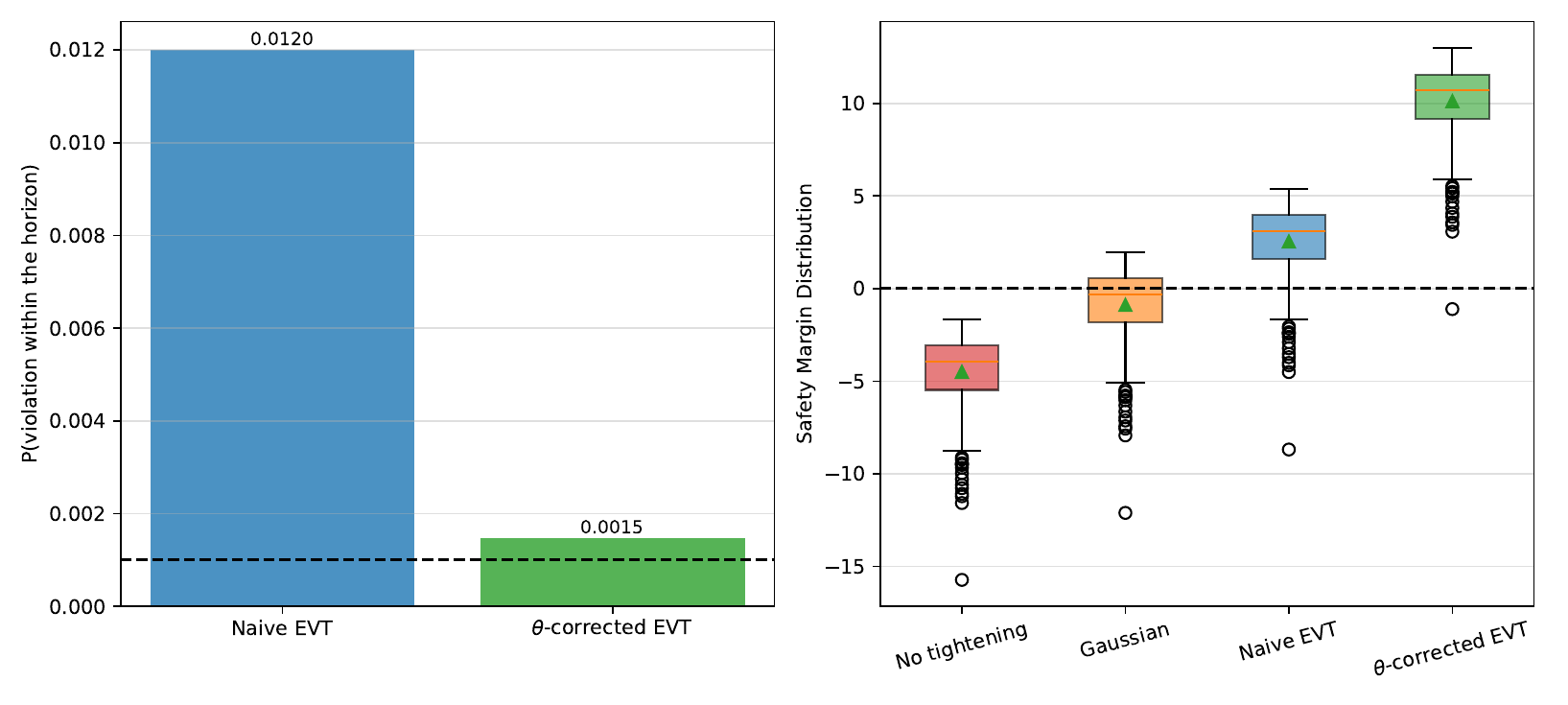}
    \caption{Effect of the cluster correction for the two-CSTR process. Left: empirical violation probability over the horizon for the naive and $\theta$-corrected EVT controllers against the target $\epsilon=10^{-3}$; Right: safety margin distribution for all four controllers over windows of $120$ time steps.}
    \label{fig_cstr_prediction}
\end{figure}
\FloatBarrier

\subsection{Nonlinear unicycle robot}\label{sec_sim_unicycle}
We next consider a nonlinear unicycle robot navigating past a circular obstacle under heavy-tailed position and heading disturbances. The robot state is $x_k=(x_{k,1},x_{k,2},\psi_k)$, with $(x_{k,1},x_{k,2})$ the planar position and $\psi_k$ the heading, and input $u_k=(v_k,\omega_k)$, evolving as
\begin{equation*}
x_{k+1} = f(x_k,u_k) + w_k, \  f(x_k,u_k) = \begin{pmatrix} x_{k,1}+\Delta t\,v_k\cos\psi_k \\ x_{k,2}+\Delta t\,v_k\sin\psi_k \\ \psi_k+\Delta t\,\omega_k \end{pmatrix},
\end{equation*}
with sampling time $\Delta t = 0.1$\,s and input saturation $|v_k| \le 1$\,m/s and $|\omega_k| \le 1.5$\,rad/s. The disturbance $w_k$ acts additively on position and heading, with each component drawn i.i.d.\ from a Student-$t_\nu$ distribution with $\nu=3$ degrees of freedom, scaled by $0.02$\,m for position and $0.01$\,rad for heading. The safety requirement is to keep the clearance from the obstacle of radius $r_o = 0.6$\,m above $d_{\text{safe}} = 0.8$\,m. The MPC prediction horizon is $N=12$.

Since Sections~\ref{sec_rare_event_SMPC} and~\ref{sec_clustering} are developed for linear dynamics, the formulation is applied to the unicycle through successive linearization, in the spirit of tube-based nonlinear MPC. At each time $t$, the dynamics and the obstacle constraint are linearized along the current nominal trajectory, which gives time-varying matrices $A_{i|t}$ and $B_{i|t}$, LQR gains $K_{i|t}$, and a half-space constraint that plays the role of~\eqref{eq_scalar_chance_const}. The quantiles $\hat q_{e,i}$ are then estimated from Monte Carlo rollouts of the resulting error dynamics, and the tightened quadratic program is solved and re-linearized for a few sequential convex programming iterations.

We compare four controllers. The baseline is a heading-only go-to-goal law with no obstacle awareness, which does not represent or tighten any obstacle constraint and serves only as a naive reference. The Gaussian and EVT controllers solve the same tightened tube SMPC problem and differ only in how they tighten the obstacle constraint. The Gaussian controller uses a light-tailed quantile estimate, while the EVT controller uses the generalized Pareto quantile estimate of the tube error tail developed in Section~\ref{sec_rare_event_SMPC}. The fourth controller, the $\theta$-corrected EVT controller, is studied in Section~\ref{sec_sim_theta}.
\begin{figure}[htb]
    \centering
        \includegraphics[width=0.9\linewidth]{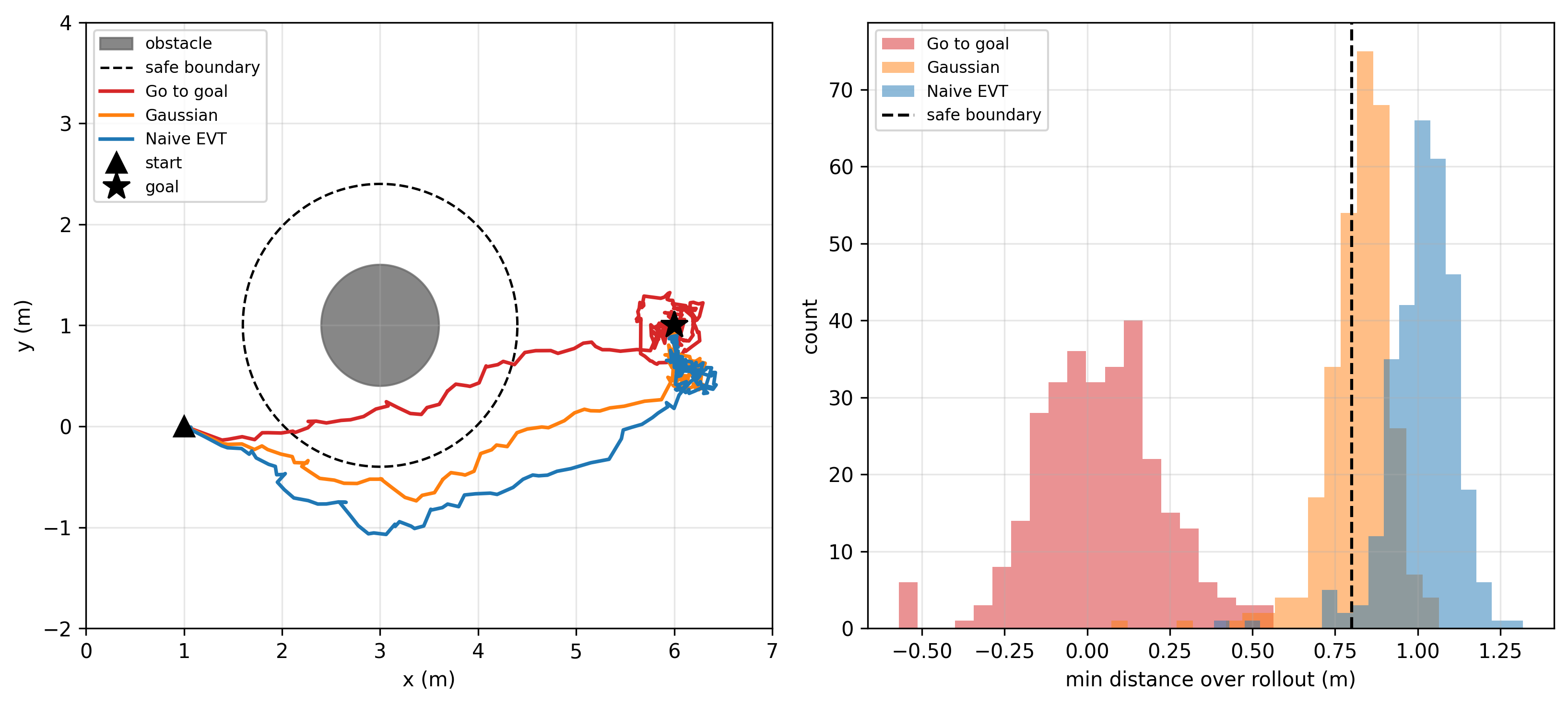}
    \caption{Left: representative trajectories; Right: Empirical distribution of the minimum distance over $300$ rollouts for the baseline, Gaussian, and naive EVT controllers.}
    \label{fig_mc_comparison}
\end{figure}

Fig.~\ref{fig_mc_comparison} shows representative trajectories of the controllers and the minimal distance to the safe boundary over 300 rollouts. Both tightened controllers solve the same tube SMPC problem with a target per-step violation probability $\epsilon=10^{-3}$. This $\epsilon$ sets how much each controller shrinks the safe region.
 
\FloatBarrier

For the three controllers, we compare their performance by mean margin - the minimum distance to the safe boundary, averaged among all rollouts, and violation rate - the proportion of rollouts that breaches the safe boundary. 
It is clear that naive EVT outperforms among the three controllers with a mean margin $0.22$\,m and a violation rate of $3.0\%$. The baseline controller has a negative mean margin ($-0.77$\,m) and violates in all trials ($100\%$), confirming it is unsafe without any tightening. Whereas, the Gaussian fitted controller attains a smaller mean margin ($0.02$\,m, closest to the $d_{\text{safe}}$ boundary) but a violation rate of $32.3\%$. A smaller margin might look more efficient, but here it instead signals under-tightening. The Gaussian quantile assumes light tails, and hence underestimates the true heavy-tailed disturbance and allows the controller approach closer to the boundary, resulting in a much higher violation rate than EVT. The worst-case margin also improves monotonically across the three, from $-1.37$\,m (baseline) to $-0.73$\,m (Gaussian) to $-0.42$\,m (Naive EVT). Note that the gap between the small target $\epsilon=10^{-3}$ and the much larger realized violation rates is expected, since $\epsilon$ only bounds the marginal per-step violation probability at a single prediction step, instead of the probability of violating over the full rollout, and hence does not eliminate rare excursions.

\begin{figure}[htb]
    \centering
    \includegraphics[width=0.6\linewidth]{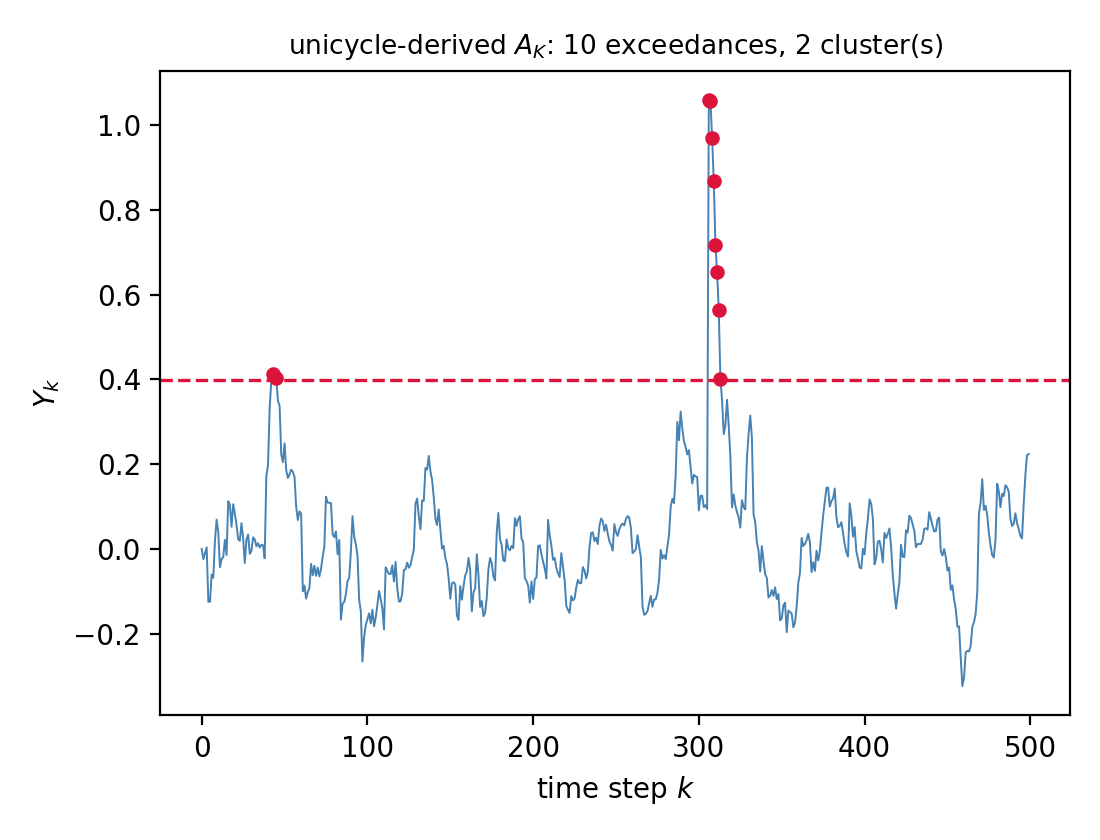}
    \caption{Extremal-index clustering on the unicycle's own closed-loop gain. A representative window showing exceedances collapsing into clusters.}
    \label{fig_theta_unicycle}
\end{figure}

\subsection{Cluster correction on the unicycle}\label{sec_sim_theta}
Figure~\ref{fig_theta_unicycle} illustrates the violation clustering directly on a representative trajectory under the unicycle example. This is precisely the clustering behavior the extremal index $\theta$ is designed to capture:  closed-loop dynamics concentrate exceedances into short clusters. A naive EVT tightening that treats exceedances as independent understates how often a rare-event episode occurs.

We now carry out the $\theta$-corrected constraint~\eqref{eq_theta_corrected_constraint} on the same closed-loop system. For this scenario, the extremal index estimated directly from the closed-loop pipeline is $\hat\theta=0.248$, which for the $N=12$ prediction horizon corrects the per-step target from $\epsilon=10^{-3}$ down to $\epsilon_\star=3.36\times10^{-4}$.
\begin{figure}[htb]
    \centering
        \includegraphics[width=0.9\linewidth]{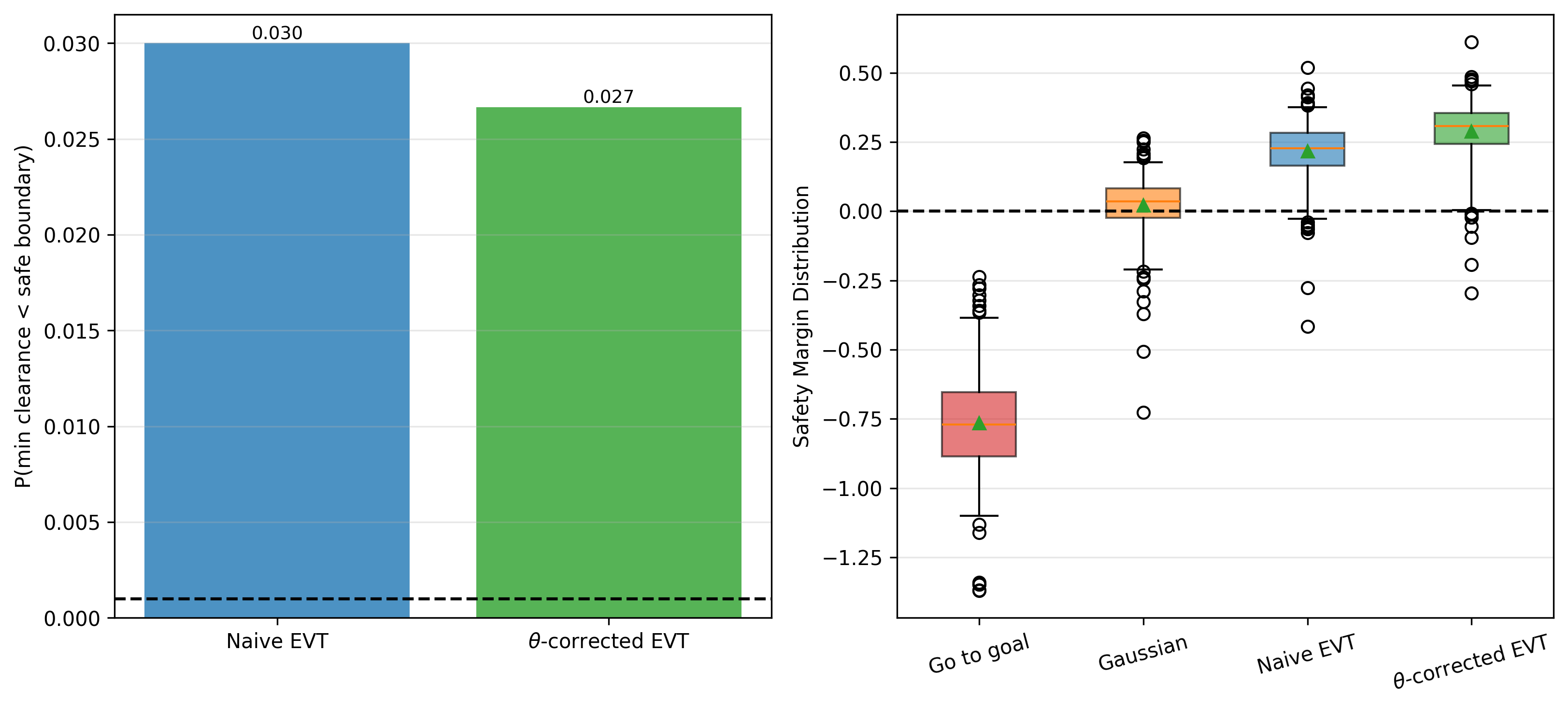}
    \caption{Effect of the cluster correction within the same tube-SMPC pipeline. Left: empirical violation probability for the naive and $\theta$-corrected EVT controllers against the target $\epsilon$; Right: safety margin distribution for all four controllers ($M=300$).}
    \label{fig_pipeline_theta}
\end{figure}
\FloatBarrier

Within this same closed-loop pipeline, the $\theta$-corrected EVT controller reduces the violation rate further, from $3.0\%$ for the naive EVT to $2.7\%$, and widens the mean margin from $0.22$\,m to $0.29$\,m, with the deepest observed excursion also shrinking from $-0.42$\,m to $-0.30$\,m. The improvement is consistent in direction with the cluster correction developed in Section~\ref{sec_clustering}, and neither EVT variant reaches the marginal target $\epsilon=10^{-3}$ in closed loop. This is expected. At each step, the controller re-linearizes the dynamics and re-fits the GPD tail quantile from a fresh, finite batch of Monte Carlo samples, which introduces finite-sample estimation error into the tightening independently of any clustering effect. The $\theta$-correction only shifts the target quantile level fed into this same estimator, so it removes the clustering bias but leaves the estimator's finite-sample error untouched, and a residual gap from the target $\epsilon$ remains regardless. The absolute violation rates in the two studies are therefore not expected to achieve $\epsilon=10^{-3}$ numerically. The shared qualitative conclusion is that uncorrected naive EVT tightening does not account for the clustering effect, and that the cluster correction lowers the realized violation without changing the controller structure.

\section{Conclusion}
This paper proposed a novel stochastic MPC formulation under heavy-tailed disturbances where the tool of EVT is introduced for the first time. An EVT-based tail characterization of the tube error, addressing the case where the disturbance distribution is unknown and heavy-tailed, in place of the invariant-set or moment-based tightening used in classical robust and stochastic MPC is proposed. The resulting quantile is asymptotically exact as the target violation probability shrinks, while the online optimization retains the computational structure of nominal linear MPC - a quadratic program. Furthermore, a closed-form extremal index that closed-loop dynamics induce even under i.i.d.\ disturbances is derived to address the temporal clustering of rare excursions. This yields a $\theta$-corrected constraint that bounds the probability of a rare-event episode over the full horizon, rather than only the marginal per-step exceedance probability.
Simulations on the two-CSTR process and a nonlinear unicycle under Student-$t$ disturbances validated both. The EVT-tightened controller achieved a substantially lower realized violation rate than the Gaussian-tightened controller by providing a larger safety margin appropriate to the heavier tail, and the $\theta$-corrected design further reduced the violation rate and widened that margin, with the observed violation-episode length in the isolated clustering validation matching the predicted clustering index.

\vspace{-3.5mm}
\section*{Code availability}
The code needed to reproduce the simulations is available at \href{https://github.com/XiuzhenYe/Stochastic-MPC-under-Heavy-Tailed-Disturbances-An-Extreme-Value-Theory-Approach}{this GitHub repository}.
\vspace{-3.5mm}
\section*{Acknowledgments}
Funding: This work was supported by the National Science Foundation (CBET-2414369).
\vspace{-3.5mm}
\section*{CRediT authorship contribution statement}
\textbf{Xiuzhen Ye:} Formal analysis, Investigation, Methodology, Software, Writing -- original draft, Writing -- review \& editing.
\textbf{Wentao Tang:} Conceptualization, Funding acquisition, Methodology, Project administration, Writing -- review \& editing.
\vspace{-3.5mm}
\section*{Declaration of generative AI and AI-assisted technologies in the manuscript preparation process}
During the preparation of this work, the authors used AI in order to adapt the simulation code of the nonlinear unicycle case to the linear two-CSTR-in-series process. The authors reviewed and edited AI-assisted code and take full responsibility for it.
\vspace{-3.5mm}
\bibliographystyle{elsarticle-num}
\bibliography{bib}

@book{Durrett2019,
  title={Probability: Theory and Examples},
  author={Durrett, Rick},
  edition={5},
  series={Cambridge Series in Statistical and Probabilistic Mathematics},
  year={2019},
  publisher={Cambridge University Press}
}

@article{McAllisterRawlings2022,
  title={Nonlinear stochastic model predictive control: Existence, measurability, and stochastic asymptotic stability},
  author={McAllister, Robert D and Rawlings, James B},
  journal={IEEE Transactions on Automatic Control},
  volume={68},
  number={3},
  pages={1524--1536},
  year={2022},
  publisher={IEEE}
}

@article{Lorenzen2017,
  title={Constraint-tightening and stability in stochastic model predictive control},
  author={Lorenzen, Matthias and Dabbene, Fabrizio and Tempo, Roberto and Allg{\"o}wer, Frank},
  journal={IEEE Transactions on Automatic Control},
  volume={62},
  number={7},
  pages={3165--3177},
  year={2017},
  publisher={IEEE}
}

@article{HewingZeilinger2020,
  title={Scenario-based probabilistic reachable sets for recursively feasible stochastic model predictive control},
  author={Hewing, Lukas and Zeilinger, Melanie N},
  journal={IEEE Control Systems Letters},
  volume={4},
  number={2},
  pages={450--455},
  year={2020},
  publisher={IEEE}
}

@article{KohlerZeilinger2025,
  title={Predictive control for nonlinear stochastic systems: Closed-loop guarantees with unbounded noise},
  author={K{\"o}hler, Johannes and Zeilinger, Melanie N},
  journal={IEEE Transactions on Automatic Control},
  volume={70},
  number={11},
  pages={7382--7397},
  year={2025},
  publisher={IEEE}
}

@article{ArsenaultChapman2022,
  title={Toward scalable risk analysis for stochastic systems using extreme value theory},
  author={Arsenault, Evan and Wang, Yuheng and Chapman, Margaret P},
  journal={IEEE Control Systems Letters},
  volume={6},
  pages={3391--3396},
  year={2022},
  publisher={IEEE}
}

@article{SomayajiLi2024,
  title={Extreme risk mitigation in reinforcement learning using extreme value theory},
  author={Somayaji N S, Karthik and Wang, Yu and Schram, Malachi and Drgona, Jan and Halappanavar, Mahantesh and Liu, Frank and Li, Peng},
  journal={Transactions on Machine Learning Research},
  year={2024}
}

@article{Mesbah_JPC_19,
  title={Mixed stochastic-deterministic tube MPC for offset-free tracking in the presence of plant-model mismatch},
  author={Paulson, Joel A and Santos, Tito LM and Mesbah, Ali},
  journal={Journal of Process Control},
  volume={83},
  pages={102--120},
  year={2019},
  publisher={Elsevier}
}

@article{TubeMPC_05,
  title={Robust model predictive control of constrained linear systems with bounded disturbances},
  author={Mayne, David Q and Seron, Mar{\'\i}a M and Rakovi{\'c}, Sa{\v{s}}a V},
  journal={Automatica},
  volume={41},
  number={2},
  pages={219--224},
  year={2005},
  publisher={Elsevier}
}

@article{davis1985limit,
  title={Limit theory for moving averages of random variables with regularly varying tail probabilities},
  author={Davis, Richard and Resnick, Sidney},
  journal={The Annals of Probability},
  volume={13},
  number={1},
  pages={179--195},
  year={1985},
  publisher={JSTOR}
}

@book{Resnick_2007,
  title={Heavy-tail phenomena: probabilistic and statistical modeling},
  author={Resnick, Sidney I},
  year={2007},
  publisher={Springer}
}

@article{hult2008tail,
  title={Tail probabilities for infinite series of regularly varying random vectors},
  author={Hult, Henrik and Samorodnitsky, Gennady},
  journal={Bernoulli},
  volume={14},
  number={3},
  pages={838--864},
  year={2008},
  publisher={Bernoulli Society}
}

@article{Mesbah2016_overview,
  title={Stochastic model predictive control: An overview and perspectives for future research},
  author={Mesbah, Ali},
  journal={IEEE Control Systems Magazine},
  volume={36},
  number={6},
  pages={30--44},
  year={2016},
  publisher={IEEE}
}

@article{PaulsonIJC17,
  title={Stochastic model predictive control with joint chance constraints},
  author={Paulson, Joel A and Buehler, Edward A and Braatz, Richard D and Mesbah, Ali},
  journal={International Journal of Control},
  volume={93},
  number={1},
  pages={126--139},
  year={2020},
  publisher={Taylor \& Francis}
}

@article{KumarZavala2019,
  title={Hierarchical MPC schemes for periodic systems using stochastic programming},
  author={Kumar, Ranjeet and Wenzel, Michael J and Ellis, Matthew J and ElBsat, Mohammad N and Drees, Kirk H and Zavala, Victor M},
  journal={Automatica},
  volume={107},
  pages={306--316},
  year={2019},
  publisher={Elsevier}
}

@article{JalvingShinZavala2022,
  title={A graph-based modeling abstraction for optimization: concepts and implementation in {Plasmo.jl}},
  author={Jalving, Jordan and Shin, Sungho and Zavala, Victor M},
  journal={Mathematical Programming Computation},
  volume={14},
  number={4},
  pages={699--747},
  year={2022},
  publisher={Springer}
}

@article{KimPetraZavala2019,
  title={An asynchronous bundle-trust-region method for dual decomposition of stochastic mixed-integer programming},
  author={Kim, Kibaek and Petra, Cosmin G and Zavala, Victor M},
  journal={SIAM Journal on Optimization},
  volume={29},
  number={1},
  pages={318--342},
  year={2019},
  publisher={SIAM}
}

@article{Lindemann2023,
  title={Safe planning in dynamic environments using conformal prediction},
  author={Lindemann, Lars and Cleaveland, Matthew and Shim, Gihyun and Pappas, George J},
  journal={IEEE Robotics and Automation Letters},
  volume={8},
  number={8},
  pages={5116--5123},
  year={2023},
  publisher={IEEE}
}

@article{Lindemann2024,
  title={Formal verification and control with conformal prediction},
  author={Lindemann, Lars and Zhao, Yiqi and Yu, Xiao and Pappas, George J and Deshmukh, Jyotirmoy V},
  journal={arXiv preprint arXiv:2409.00536},
  year={2024}
}

@article{FarinaGiulioniScattolini2016,
  title={Stochastic linear model predictive control with chance constraints -- a review},
  author={Farina, Marcello and Giulioni, Luca and Scattolini, Riccardo},
  journal={Journal of Process Control},
  volume={44},
  pages={53--67},
  year={2016},
  publisher={Elsevier}
}

@book{MarshallOlkin1979,
  title={Inequalities: Theory of Majorization and Its Applications},
  author={Marshall, Albert W and Olkin, Ingram},
  series={Mathematics in Science and Engineering},
  volume={143},
  year={1979},
  publisher={Academic Press},
  address={New York}
}

@article{CalafioreElGhaoui2006,
  title={On distributionally robust chance-constrained linear programs},
  author={Calafiore, Giuseppe C and El Ghaoui, Laurent},
  journal={Journal of Optimization Theory and Applications},
  volume={130},
  number={1},
  pages={1--22},
  year={2006},
  publisher={Springer}
}

@book{Coles2001,
  title={An Introduction to Statistical Modeling of Extreme Values},
  author={Coles, Stuart},
  series={Springer Series in Statistics},
  year={2001},
  publisher={Springer-Verlag},
  address={London}
}

@article{FerroSegers2003,
  title={Inference for clusters of extreme values},
  author={Ferro, Christopher A T and Segers, Johan},
  journal={Journal of the Royal Statistical Society: Series B (Statistical Methodology)},
  volume={65},
  number={2},
  pages={545--556},
  year={2003},
  publisher={Wiley}
}

@book{Leadbetter1983,
  title={Extremes and Related Properties of Random Sequences and Processes},
  author={Leadbetter, M R and Lindgren, G and Rootz{\'e}n, H},
  series={Springer Series in Statistics},
  year={1983},
  publisher={Springer-Verlag},
  address={New York}
}

@article{MayneRawlingsRaoScokaert2000,
  title={Constrained model predictive control: Stability and optimality},
  author={Mayne, David Q and Rawlings, James B and Rao, Christopher V and Scokaert, Pierre OM},
  journal={Automatica},
  volume={36},
  number={6},
  pages={789--814},
  year={2000},
  publisher={Elsevier}
}

@incollection{BemporadMorari1999,
  title={Robust model predictive control: A survey},
  author={Bemporad, Alberto and Morari, Manfred},
  booktitle={Robustness in Identification and Control},
  series={Lecture Notes in Control and Information Sciences},
  volume={245},
  pages={207--226},
  year={1999},
  publisher={Springer-Verlag},
  address={London}
}

@article{YeTang2025AIChE,
  title={Process resilience under optimal data injection attacks},
  author={Ye, Xiuzhen and Tang, Wentao},
  journal={AIChE Journal},
  volume={71},
  number={8},
  pages={e18896},
  year={2025},
  doi={10.1002/aic.18896}
}
 
\end{document}